\documentclass[letterpaper, 10 pt, conference]{ieeeconf} 

\IEEEoverridecommandlockouts                             
\usepackage{amssymb}
\usepackage{amsmath}
\usepackage{amsfonts}%
\usepackage{graphicx}
\usepackage{epstopdf}
\usepackage{color}
\usepackage{xr}%
\usepackage{subcaption}
\usepackage{cite}
\usepackage{url}
\newtheorem{theorem}{Theorem}

\newtheorem{lemma}{Lemma}

\title{\LARGE \bf
Delocalized Epidemics on Graphs: A Maximum Entropy Approach
}

\author{Faryad Darabi Sahneh$^{1,*}$, Aram Vajdi$^1$, Caterina Scoglio$^1$\thanks{This material is based upon work supported by the National Science Foundation under Grant No. CIF-1423411.}\thanks{$^1$Sahneh, Vajdi, and Scoglio are with the Department of Electrical and Computer Engineering at Kansas State University,
Manhattan, KS, 66506. Emails: {faryad,avajdi,caterina}@ksu.edu.}
\\{\bf (In Proceedings of American Control Conference, 2016)}}

\begin{document}

\maketitle

\begin{abstract}
The susceptible--infected--susceptible (SIS) epidemic process on complex networks can show metastability, resembling an endemic equilibrium. In a general setting, the metastable state may involve a large portion of the network, or it can be localized on small subgraphs of the contact network. Localized infections are not interesting because a true outbreak concerns network--wide invasion of the contact graph rather than localized infection of certain sites within the contact network. Existing approaches to localization phenomenon suffer from a major drawback:  they fully rely on the steady--state solution of mean--field approximate models in the neighborhood of their phase transition point, where their approximation accuracy is worst; as statistical physics tells us. We propose a dispersion entropy measure that quantifies the localization of infections in a generic contact graph. Formulating a maximum entropy problem, we find an upper bound for the dispersion entropy of the possible metastable state in the exact SIS process. As a result, we find sufficient conditions such that any initial infection over the network either dies out or reaches a localized metastable state. Unlike existing studies relying on the solution of mean--field approximate models, our investigation of epidemic localization is based on characteristics of exact SIS equations. Our proposed method offers a new paradigm in studying spreading processes over complex networks. 
\end{abstract}

\section{Introduction}
Epidemic models aspire to describe the spread of viruses/worms/ideas in biological/computer/social networks\cite{wang2003epidemic,daley2001epidemic,kephart1991directed}.
One of the simplest epidemic models over a complex network is the susceptible--infected--susceptible (SIS) model. In the SIS epidemic model, each individual is either infected or susceptible. An infected individual spreads the virus to her susceptible neighbors with an \textit{infection rate} $\beta>0$ and itself get cured with the \textit{curing rate} $\delta>0$ and becomes susceptible again. SIS model exemplifies a networked dynamical system where the interaction between simple node-level dynamics and network topology leads to nontrivial emergent behaviors.

Despite the simple description of the SIS process, only a few \emph{exact} results about the SIS process on a generic graph $G$ have been proposed. In the SIS process, the disease--free state is an absorbing state, i.e., any initial infection will ultimately die out regardless of the infection rate\cite{ganesh2005effect,van2009virus}. The extinction time depends on the structure of the network, the infection and curing rates, and the initial infection. For this model, Ganesh \textit{et al.} \cite{ganesh2005effect} rigorously proved that any initial infections die out exponentially in time if the infection strength, $\tau\triangleq\beta/\delta$, is smaller than the inverse of the spectral radius of the graph $\rho(G)$.  However, for the values of $\tau$ larger than $1/\rho(G)$, the process may reach \emph{metastability} where the extinction time is exponentially long with respect to the population size and the process stays in a state that resembles equilibrium \cite{mountford2013metastable}.

A true epidemic outbreak concerns network--wide invasion of the contact graph rather than localized infection of certain sites within the contact network. This argument leads us to the concept of infection \textit{localization} in the SIS model over a generic graph. To illustrate this phenomenon, consider the Line--Clique graph in Fig.(\ref{motivatingexampleth.eps}) consisting of two subgraphs, the clique part of size $m$ and the line part with size $N>>m$. The spectral radii of the clique part and the line part separately are $m-1$ and $\sim2$, respectively. However, the spectral radius of the Line--Clique graph is close to that of the clique subgraph. For such a graph, any infection dies out exponentially in time as long as the infection strength $\tau$ is smaller than $1/(m-1)$. However, what happens for $\tau>1/(m-1)$? We know for $\tau\leq1/2$, the line subgraph, considered separately, cannot sustain infections for a long time. The argument above leads to the speculation that for  $1/(m-1)\leq\tau\leq1/2$, if the Line--Clique network with $N>>m$ reaches a metastable state, the infection should be mostly localized on the clique part of the network. Such localized invasions are not interesting because they concern tiny portions of the contact network.

\begin{figure}[ptb]%
\centering
\includegraphics[width=3.4in]{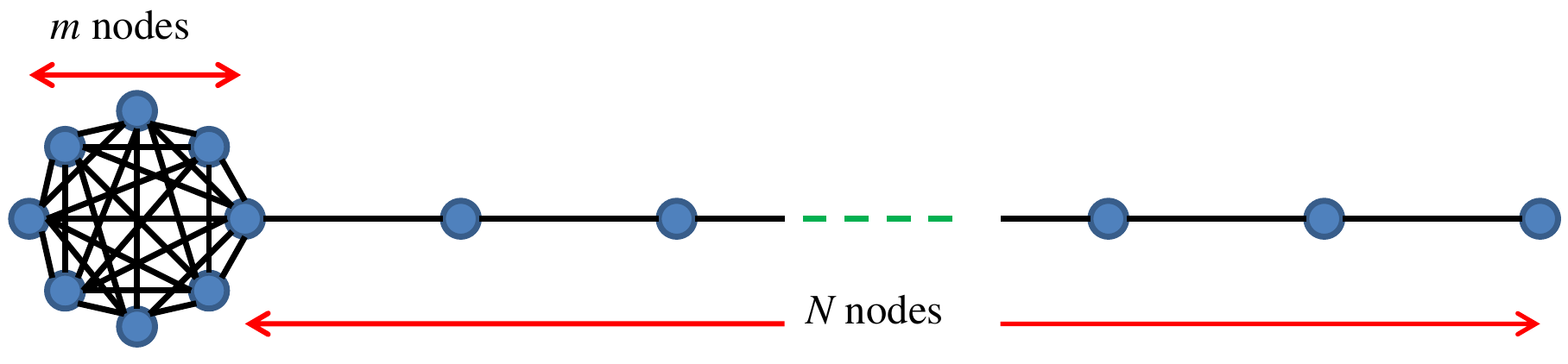}
\caption{The Line-Clique graph consisting of a complete graph of size $m$ and a line graph of size $N>>m$. It is possible to observe a metastable state where infections mostly localize on the clique part --- a tiny portion of the network.}%
\label{motivatingexampleth.eps}%
\end{figure}

Localization of SIS process has recently been reported in the literature. Goltsev \textit{et al.} \cite{goltsev2012localization} studied the steady--state solution of the mean--field approximated SIS model for $\tau$ close to $1/\rho(G)$, where the equilibrium solution is proportional to the dominant eigenvector of the contact network adjacency matrix. The major drawback of such approaches \cite{goltsev2012localization,boguna2013nature,odor2013spectral} is that they fully rely on approximate models in a region where they are least accurate. Mean--field models perform more accurately at early times and for large values of $\tau$, while they can perform very poorly at steady--state and for $\tau$ close to $1/\rho(G)$.

In this paper, we propose a dispersion measure based on Kullback–Leibler divergence \cite{kullback1951information} that quantifies how the marginal probability of infection is far from a homogeneous spread over the nodes of the network. We show that formulating a maximum entropy problem, we can find an upper bound for the dispersion entropy of the possible metastable state. As a result, any initial infection over the network either dies out or reaches a metastable state that has lower entropy than the upper bound. Unlike existing studies, our investigation of epidemic localization does not use mean-field approximation of the SIS process and is based on exact equations arguments. Convex optimization techniques \cite{boyd2004convex} allows for efficient solution of the maximum entropy problem even for large networks. Numerous Monte Carlo simulation of the SIS model support our results.

The rest of the paper is organized as follows. In Section II, we briefly explain graph theory tools used for modeling the SIS process as well as the equations that govern the SIS model over a complex network. Section III discusses our approach to epidemic localization. Supporting numerical experiments are presented in Section IV. The paper concludes in Section V.

\section{SIS Epidemic Spreading Model on a Graph}

Graph theory is
widely used for representing the contact topology in an epidemic network\cite{newman2010networks, van2010graph}. In
the SIS model, the network of $N$ agents is represented by a graph
$G=\left\{  \mathcal{V},\mathcal{E}\right\}  $, where $\mathcal{V}$
is the set of agents and $\mathcal{E\subseteq V\times V}$ denotes the set of
edges between agents. An edge $(i,j)\in\mathcal{E}$ exists if agent $j$ can
directly infect agent $j$. In this paper, we assume the
contact graph is undirected, i.e., for any $(i,j)\in\mathcal{E}$, we have
$(j,i)\in\mathcal{E}$. For the contact graph $G$, the adjacency matrix
$A=\left[  a_{ij}\right]  \in \mathbb{R}^{N\times N}$ consists of the elements $a_{ij}=1$ if and only if
$(i,j)\in\mathcal{E}$ else $a_{ij}=0$. Spectral radius of a graph $G$, denoted by $\rho(G)$, is defined as the spectral radius of its adjacency matrix --- the largest magnititude of the adjacency matrix $A$ eigenvalues. Moreover, a path $\mathcal{P}$ of length $k$ between vertices $v_{0}$ and $v_{k}$ is an ordered sequence $(v_{0},...,v_{k})$
where $(v_{i-1},v_{i})\in\mathcal{E}$ for $i=1,...,k$.  Graph $G%
$\ is (strongly) connected if any two vertices are connected with a path. If the contact graph is connected, the adjacency matrix $A$ is irreducible and according to the Perron--Frobenius theorem \cite{van2010graph}, the largest eigenvalue of $A$, denoted by $\lambda_{1}(A)$, is a positive real number and the corresponding eigenvector $\boldsymbol{x}_1(A)$ is the only eigenvector of the adjacency matrix with all positive elements.

In the SIS\ model, the state $X_{i}(t)$ of an agent $i$ at time $t$ is a
Bernoulli random variable, where $X_{i}(t)=0$ if agent $i$ is susceptible and
$X_{i}(t)=1$ if it is infected. The \emph{curing process} for the infected agent
$i$ has an exponentially distributed time duration described by the curing rate $\delta\in\mathbb{R}^{+}$. 
The \emph{infection} process for the susceptible agent $i$ in contact with only one infected neighbor has an exponentially distributed time period characterized by the infection rate $\beta\in\mathbb{R}^{+}$. An agent in contact with more than one infected neighbors occurs at rate $\beta Y_{i}(t)$,
where $Y_{i}(t)\triangleq\sum_{j=1}^{N}a_{ij}X_{j}(t)$ is the number of
infected neighbors of agent $i$ at the time $t$. Fig. \ref{sis.eps} shows a schematic of the
SIS epidemic spreading model over a generic graph.

\begin{figure}[ptb]%
\centering
\includegraphics[width=3.3in]{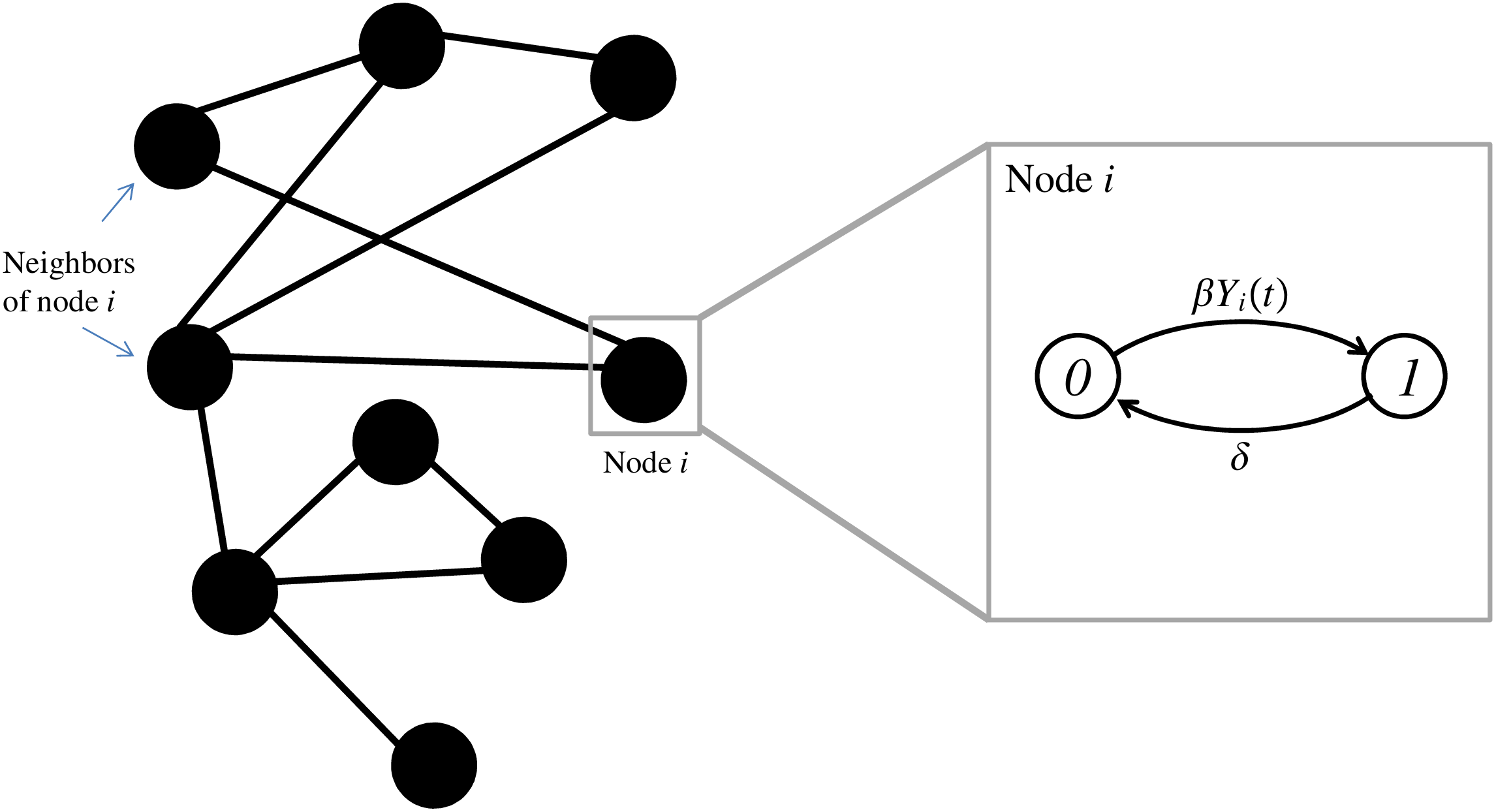}
\caption{Schematics of a contact network along with the agent-level stochastic transition diagram for agent $i$ according to the SIS epidemic spreading model. The parameters $\beta$ and $\delta$ denote the infection rate and curing rate, respectively. $Y_{i}(t)$ is the number of infected neighbors of agent $i$ at time $t$.}%
\label{sis.eps}%
\end{figure}

The networked SIS process, described above, is indeed a Markov process on the network state $\mathbf{X}\triangleq[X_1,...,X_N]^T$ with the state--space of $2^N$, i.e., next event on the network state only depends on the current network state. The node--level description of the SIS process is
\begin{align*}
&Pr(X_{i}(t+\Delta t)=1|X_{i}(t)=0,\mathbf{X}(t))=\beta Y_{i}(t)\Delta
t+o(\Delta t) \\
&Pr(X_{i}(t+\Delta t)=0|X_{i}(t)=1,\mathbf{X}(t))=\delta\Delta t+o(\Delta t),
\end{align*}
where $o(\Delta t)$ denotes terms such that $lim_{\Delta t\rightarrow 0}\frac{o(\Delta t)}{\Delta t}=0$. From the node--level description of the SIS process and following the procedure of \cite{sahneh2013generalized}, the differential equations governing the expected value $E[X_{i}]$ can be written as
\begin{align}
\frac{d}{dt}E[X_{i}] &  =\beta\sum a_{ij}E[(1-X_{i})X_{j}]-\delta E[X_{i}]\label{exact}\\
\nonumber
&  =\beta\sum a_{ij}E[X_{j}]-\beta\sum a_{ij}E[X_{i}X_{j}]-\delta E[X_{i}]\label{exact}
\end{align}
for $i \in \{1,...,N\}$ \cite{cator2012second}. Eq. (\ref{exact}) is not a closed system as the evolution of $E[X_i]$ depends on the joint probabilities of the pairs $X_{i}X_{j}$. Furthermore, if we proceed to derive the time derivative of $E[X_iX_j]$, it turns out dependent on higher order terms  $E[X_iX_jX_k]$ which are expected values of triplets. The procedure goes on until we reach a closed system of $2^N-1$ equations involving $E[X_i...X_N]$. Such exponentially enormous state space of the exact model challenges feasibility of analytical investigation of exact SIS process.

Only a few \emph{exact} results currently exist for the SIS process on a generic graph $G$. In the SIS process, the disease--free state is an absorbing state, i.e., any initial infection will ultimately die out regardless of the infection rate\cite{ganesh2005effect,van2009virus}. The extinction time, in general, depends on the snetwork tructure network, infection and curing rates, and the initial infection. Ganesh et al.\cite{ganesh2005effect} rigorously proved that  for $\tau<1/\rho(G)$, any initial infections die out exponentially. In particular:

\begin{theorem}[Ganesh et al. \cite{ganesh2005effect}] \label{Th: Ganesh}
For $\tau<1/\rho(G)$, the probability that initial infections $X(0)$ have not completely died out by time $t$ is upper-bounded by an exponentially decaying function as
\[
\Pr(\mathbf{X}(t)\neq0)\leq\sqrt{N \textstyle\sum_{i=1}^{N}X_{i}(0)}\ e^{(\tau\rho(G)-1)t}.
\]
\end{theorem}

Theorem \ref{Th: Ganesh} suggests a sufficient condition for the complete extinction of infection in the graph. However, considering the example discussed in the introduction section, for $\tau>1/\rho(G)$, initial infection may reach a metastable state localized in a small number of nodes. For such cases achieving a metastable state does not indicate an epidemic outbreak. In this paper, we go beyond the result of  Ganesh et al. \cite{ganesh2005effect} and seek conditions that lead to either complete extinction of infections or their localized persistence.

\section{Existing Approaches to Localization}

A commonly adopted technique to overcome the complexity of exact SIS equations (\ref{exact}) is moment closure, where high order expectations are \emph{approximated} by lower order terms. In particular, a first-order closure approximation, also referred to as mean-field type approximation, assumes that $X_{i}$ and $X_{j}$ are uncorrelated, i.e., $E[X_{i}X_{j}]\approx E[X_{i}]E[X_{j}]$. The mean-field approximation leads to a system of nonlinear differential equations as \cite{van2009virus}
\begin{equation}
\dot{v}_{i}=\beta(1-v_{i})\sum a_{ij}v_{j}-\delta v_{i},\label{NIMFA}%
\end{equation}
known as N-intertwined mean-field approximation (NIMFA) model; extensively studied in the literature \cite{van2011n,van2012n,van2012epidemic,li2012susceptible,khanafer2014stability}. Van Mieghem and Van de Bovenkamp rigorously proved that the solution of NIMFA model (\ref{NIMFA}) upper-bounds the marginal infection probability $p_i\triangleq E[X_i]$ from exact equations (\ref{exact}), i.e., $\forall i,\  p_i(t)\leq v_i(t)$ if $p_i(0)=v_i(0)$ \cite{van2013non}. Furthermore, the equilibrium points of NIMFA model (\ref{NIMFA}) satisfying \cite{van2009virus}
\begin{equation}
\frac{v_{i}^{\ast}}{1-v_{i}^{\ast}}=\tau\sum a_{ij}v_{j}^{\ast},
\label{NIMFA_eq}
\end{equation}
shows a bifurcation behavior at $\tau_c=1/\rho(G)$. Specifically, for $\tau<1/\rho(G)$, the disease-free state $v^*_i=0$ is the only equilibrium point while for $\tau>1/\rho(G)$ there exists a second equilibrium point with $v^*_i>0, \forall i \in \{1,...,N\}$ \cite{van2009virus}. Furthermore, for $\tau$ very close to $1/\rho(G)$, the positive equilibrium point $\boldsymbol{v}^*=[v^*_1,...,v^*_N]^T$ is parallel to the dominant eigenvector of the adjacency matrix $A$, i.e., $\boldsymbol{v}^*=c\tilde{\tau}\boldsymbol{x}_1(A)+o(\tilde{\tau})$ for $\tilde{\tau}\triangleq \tau-1/\rho(G)$ and some constant $c$ \cite{van2009virus}.

Considering the relation between the steady state solution of NIMFA model close to the critical point $\tau_c=1/\rho(G)$ and dominant eigenvector of the adjacency matrix, Goltsev \textit{et al.} \cite{goltsev2012localization} studied the homogeneity of the dominant eigenvector in order to address the localization of infection. Specifically, they used the inverse participation ratio (IPR) defined as
\begin{equation}
IPR\triangleq \frac{1}{N\sum \boldsymbol{x}^4_{1,i}},
\label{IRP}
\end{equation}
as a measure of localization; claiming that $IRP>0$ indicates delocalized epidemics while $IRP\rightarrow 0$ pin-points a localized one.   However, as we argued earlier, such approaches---as adopted in \cite{goltsev2012localization,boguna2013nature,odor2013spectral}---has the unjustified reliance on mean-field approximate model, at the steady state, and close to the bifurcation transition point $\tau_c=1/\rho(G)$. Statistical mechanics tells us mean-field models perform worst in the neigborhood of bifurcation transition points. Indeed, their key equation $\boldsymbol{v}^*=c\tilde{\tau}\boldsymbol{x}_1(A)+o(\tilde{\tau})$ can be totally irrelevant to the actual meta-stable state distribution.

In this paper, instead of \textit{assuming} a (possibly highly inaccurate) approximate \textit{solution} to the meta-stable state, we focus on a simple \textit{property} of the possible meta-stable state, directly drived from \textit{exact} equations.

\section{Main Results}
We are particularly interested in study of infection delocolization in the metastable state. Metastable state resembles an equilibrium where the infection probability of each node stays (almost) constant, i.e., $d\boldsymbol{p}(t)/dt\rightarrow0$, where $\boldsymbol{p}=[p_1,...,p_N]^T$. To begin with, we use a simple observation from the exact SIS equations (\ref{exact}) that $dp_i/dt\leq \beta\sum a_{ij}p_j-\delta p_i$, due to the fact that $X_iX_j$, and as a consequence, $E[X_iX_j]$ is nonnegative. Therefore, when metastability is achieved, we must have
\begin{equation}
(\beta A-\delta I)\boldsymbol{p}_\infty
\geq 0,\label{a2}
\end{equation}
where, by slight abuse of notation, $\boldsymbol{p}_\infty$ denotes marginal infection probability vector in the metastable state.

We do not intend to compute $\boldsymbol{p}_\infty$. Instead, we look at the set of all plausible positive vectors that satisfy (\ref{a2}). The key idea is that if all the vectors in this set are localized then necessarily the actual metastable state is also localized. For this, we need to define a measure of delocalization.

We propose to use utilize the notion of distance between probability distributions to develop a delocalization measure. First, note that summation of marginal infection probabilities $\sum_{i}^{N}p_{i}$ provides a descriptor for the expected size of the epidemic, however, it does not provide any information on how the infection is distributed among the agents. Therefore, in order to study the dispersion of infection regardless its size, first we can normalize infection probabilities $p_i$ by $\sum_{i}^{N}p_{i}$, and work with $\overline{p}_{i}\triangleq p_{i}/\sum_{i}^{N}p_{i}$. Since $\sum_{i}^{N}\overline{p}_{i}=1$, we treat $\overline{\boldsymbol{p}}=[\overline{p}_{1},...,\overline{p}_{N}]^T$ as a probability distribution and then utilize concepts of distance between distributions to quantify the distance between $\overline{\boldsymbol{p}}$ and uniform distribution $\boldsymbol{u}=[\frac{1}{N},...,\frac{1}{N}]^T$. In particular, we use Kullback-Leibler divergence \cite{kullback1951information} which measures
\begin{equation}
D_{KL}(\overline{\boldsymbol{p}}||\boldsymbol{u})=-\sum_{i}^{N}\overline{p}_{i}\ln{\overline{p}_{i}}-\ln(N).
\end{equation}
Therefore, in order to study the level of delocalization, we can use entropy as a measure. For any probability distribution of infection $p_{i}$, we can calculate a \textit{dispersion entropy} as
\begin{equation}
S(\boldsymbol{p})=-\sum_{i}^{N}\overline{p}_{i}\ln{\overline{p}_{i}},\label{a4}
\end{equation}
where $\overline{p}_{i}=p_{i}/\sum_{i}^{N}p_{i}$ and $N$ is the number of nodes in the network. The defined entropy reaches its maximum, $\ln(N)$, when $D_{KL}(\overline{\boldsymbol{p}}||\boldsymbol{u})=0$, i.e., all the nodes have same none-zero probability of infection.

\begin{theorem} \label{Th: MEP}
Assuming a contact graph $G$, infection strength $\tau$ and initial infection probability $\boldsymbol{p}(0)$, if a metastable state is achieved, the dispersion entropy of the metastable state is upper-bounded by $S^*$ which is the solution  of the following maximum entropy problem:
\[
maximize: \  \ \ \ \ S=-\textstyle\sum_{i}^{N}p_{i}\ln{p_{i}} ,
\]
\[
subject\ to: \ \ (\tau A-I)\boldsymbol{p}\geq 0,\\
\]
\[
\ \ \ \ \ \ \ \ \ \ \ \ \textstyle \sum_{i}^{N}p_{i}=1, \label{MEP}\\
\]
\[
\ \ \ \ \ \ \boldsymbol{p}>0.
\]

\end{theorem}

\begin{proof}
The solution to the above optimization problem maximizes the entropy defined in Eq. (\ref{a4}) because, instead of normalization of the probabilities, $\sum_{i}^{N}p_{i}=1$ has been added to the constraints set and inequality (\ref{a2}) is linear which is not altered by scaling $\boldsymbol{p}$.
\end{proof}

The maximum entropy problem can be solved efficiently for large network sizes using convex optimization techniques \cite{boyd2004convex}.

\begin{lemma}\label{lmfe}
If $\tau<1 /\rho(G)$, there does not exist any $\boldsymbol{p}$ that satisfies condition (\ref{a2}). Furthermore, for $\tau\geqslant 1 /\rho(G)$, the constraint of Theorem 2 has a non-empty feasible set.
\end{lemma}

\begin{proof}
Any feasible probability distribution $\boldsymbol{p}>0$ that satisfies condition (\ref{a2}) satisfies 
\begin{equation}
\boldsymbol{p}^{T}(\tau A-I)\boldsymbol{p}\geq 0.\label{a3}
\end{equation}
However, if $\tau<1 /\rho(G)$, matrix $(\tau A-I)$ is a negative definite matrix which cannot allow (\ref{a3}). Therefore, if $\tau<1 /\rho(G)$, there does not exist any $\boldsymbol{p}$ that satisfies condition (\ref{a2}). On the other hand, for $\tau\geqslant 1 /\rho(G)$, the dominant eigenvector of $A$, i.e., $\boldsymbol{p}=\frac{1}{||\boldsymbol{x}_1(A)||_1}\boldsymbol{x}_1(A)$, is always feasible.
\end{proof}

We would like to remark the existence of a distribution with a high value of dispersion entropy that satisfies condition (\ref{a2}) does not indicate the existence of a metastable state. However, \emph{if there exist a metastable state}, our analysis assigns an upper bound to its dispersion entropy. Hence, if the optimization problem yields a small value for entropy, the infection does not invade a large number of nodes in the metastable state, hence providing a \emph{sufficient condition for either complete extinction of infections or their localized persistence}.  
  
Moreover, if $\tau\downarrow 1/\rho(G)$, Lemma \ref{lmfe} indicates the feasible space of optimization problem is a small neighborhood including the dominant eigenvector of the adjacency matrix $\boldsymbol{x}_1(A)$, which makes $S^*\simeq S(\boldsymbol{x}_1(A))$. In this case, the results of our analysis are compatible with those of \cite{goltsev2012localization}, except we use a different measure for localization. However for higher values of $\tau$, our analysis can still provide an upper bound for the delocalization of SIS process; while an analysis based on the mean-field approximation does not necessarily characterize the infection delocalization in exact SIS process.
 
\begin{figure}[ptb]
\begin{subfigure}{3.4in}
\includegraphics[scale=0.39]{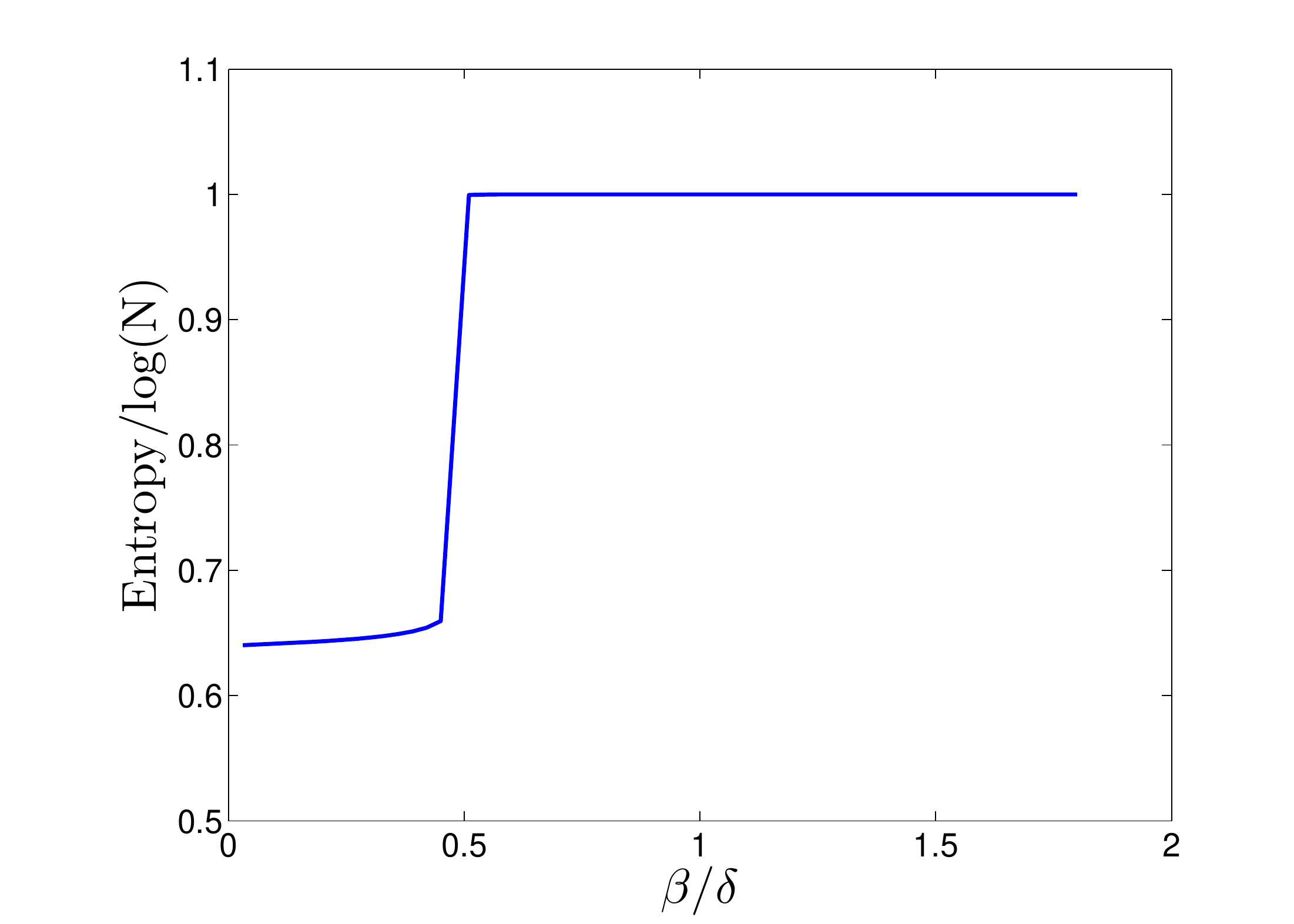}
\caption{}
\label{lcop}
\end{subfigure} 
\begin{subfigure}{3.4in}
\includegraphics[scale=0.51]{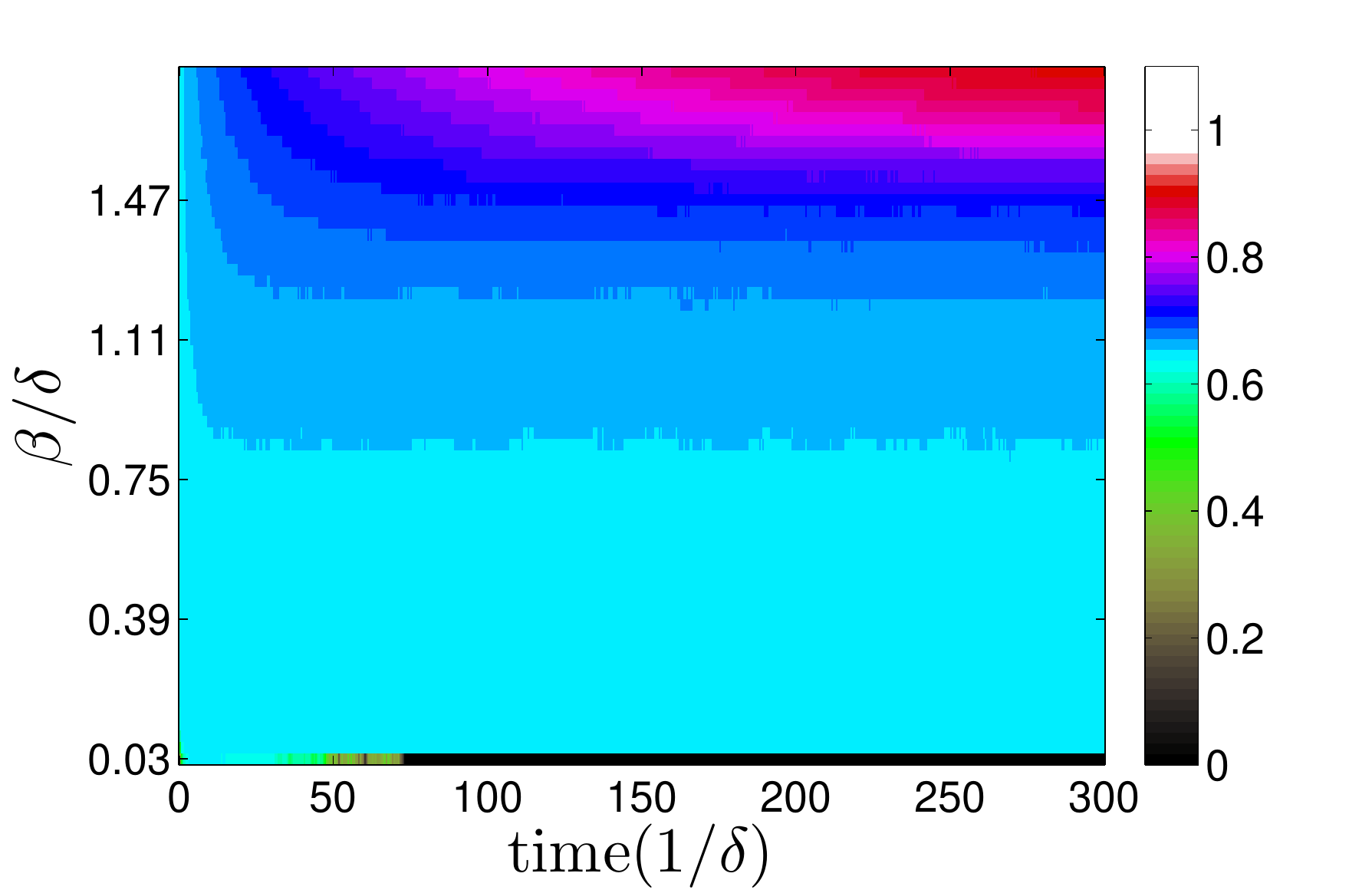}
\caption{}
\label{lcsim}
\end{subfigure}
\caption{(a) The entropy of the optimized distribution for the Line-Clique graph in Fig. \ref{motivatingexampleth.eps}. As can be seen, there is a sudden jump at $\tau=\frac{1}{2}$. (b) Monte Carlo simulation of the SIS model over the Line-Clique graph. Color represents dispersion entropy of infection probability distribution divided by $\ln(N)$.}
\label{lc}
\end{figure}

\begin{figure}[ptb]
\begin{subfigure}{3.40in}
\includegraphics[scale=.43]{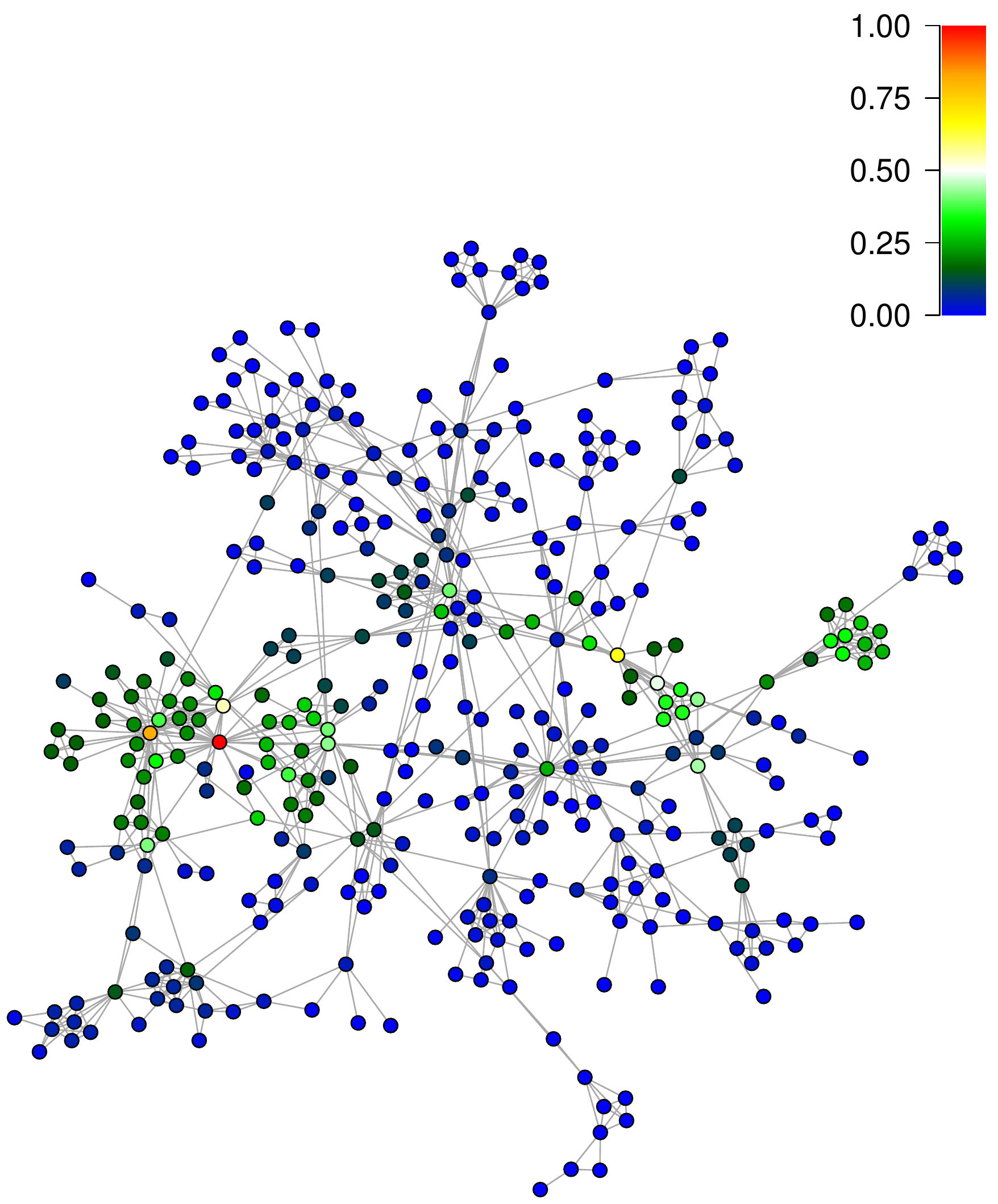}
\caption{}
\label{arba}%
\end{subfigure} 
\begin{subfigure}{3.4in}
\includegraphics[scale=.43]{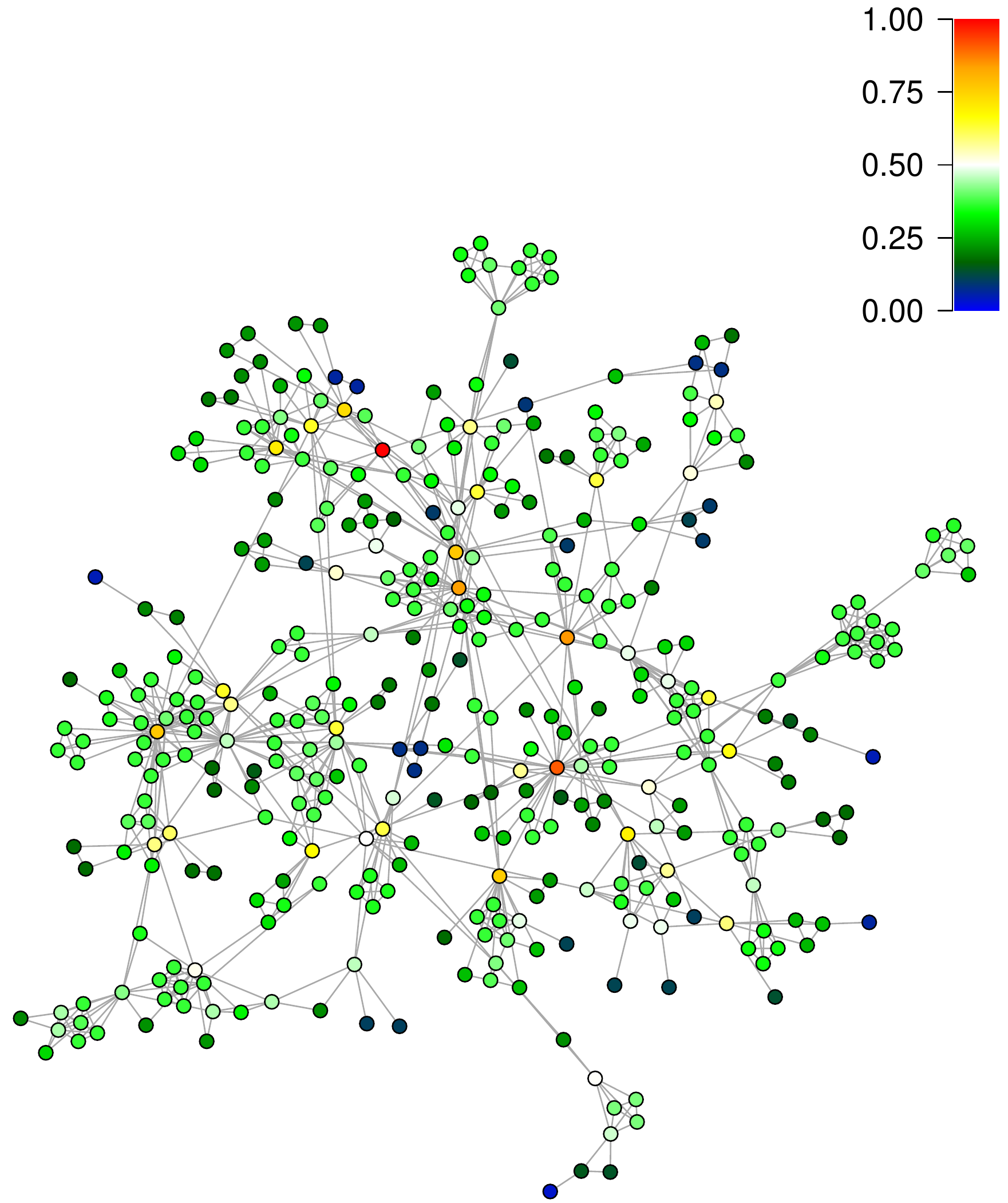}
\caption{}
\label{arbb}%
\end{subfigure}
\caption{(a) Optimized probability distribution for  $\beta/\delta=0.125$, showing only a few localized sites of the network have active nodes (b) Optimized probability distribution for $\beta/\delta=0.23$.}
\label{arb1}%
\end{figure}
\section{numerical result}
Considering the toy graph in Fig. \ref{motivatingexampleth.eps}, we generated a Line-Clique graph with 280 nodes in the line subgraph and 40 nodes in the clique subgraph. We used the CVX package \cite{cvx} --- a Matlab-based modeling system for deciplined convex optimization --- to compute the maximum entropy corresponding to any infection strength of interest. Fig. \ref{lcop} depicts the computed upper bound for $\tau>\frac{1}{\rho(G)}$. For $\beta/\delta>\frac{1}{2}$, the upper-bound values for the entropy of infection distribution in possible metastable is close to $\ln(N)$, which is indeed a trivial upper bound. However, for $\beta/\delta<\frac{1}{2}$, the upper-bound values are much smaller than the entropy of homogeneous distribution --- $\ln(N)$ --- indicating that if the epidemic reaches metastability, the infection will not spread to the whole nodes of the network and instead will be localized on a site of size $\sim(280+40)^{0.64}\approx 40$ nodes at most; which, interestingly, is the size of the clique part.       
Moreover, to show the relation between the computed upper bound and the true entropy of infection,  we performed Monte Carlo simulation of the SIS model over the Line-Clique graph using GEMFsim package \cite{GEMFsim} --- a Gillespie-based simulator for the generalized epidemic modeling framework in \cite{sahneh2013generalized}. For this simulation, we assumed an initial condition where only one node in the clique subgraph was infected. Fig. \ref{lcsim} shows the result of the simulation, where color represents the true dispersion entropy divided by $\ln(N)$ through time and as function of infection strength. For   $\beta/\delta<\frac{1}{2}$, the dispersion entropy of infection grows very fast and stays constant with values less than the computed upper bound in Fig. \ref{lcop}. In this case, epidemics does reach metastability, however, infections are localized on the clique subgraph.

As another example, we chose the largest component of a coauthorship network from \cite{newman2006finding} as shown in Fig. \ref{arb1}. For this network, the spectral radius of the adjacency matrix is $\sim10.4$. The entropy of the optimized distribution, shown in Fig. \ref{arbop}, is an upper bound for the metastable state of SIS model over the network. Even though for $0.13<\beta/\delta<0.15$ the optimized entropy has a large value, we cannot predict existence of a metastable state. In fact, the result of Monte Carlo simulation (Fig. \ref{arbsim}) for SIS model over the network shows the metastable state starts at much higher values of $\beta/\delta$ (larger than 0.2) where the optimized entropy is almost $\ln(N)$.
     
Moreover, as an illustration for the relationship between the dispersion entropy of a distribution and delocalization of the distribution, we have plotted the network and colored the nodes based on the value of its probability in the optimized distribution. In Fig. \ref{arba} and Fig. \ref{arbb} , the optimized distribution for two different values of $\beta/\delta$ is plotted. For  $\beta/\delta=0.125$, where the optimized entropy is small, the distribution is mainly localized on a few nodes. On the other hand, when $\beta/\delta$ increases to $0.23$, the entropy of the optimized distributions increases and more nodes get involved.

 \begin{figure}[h]
\begin{subfigure}{3.4in}
\includegraphics[scale=.4]{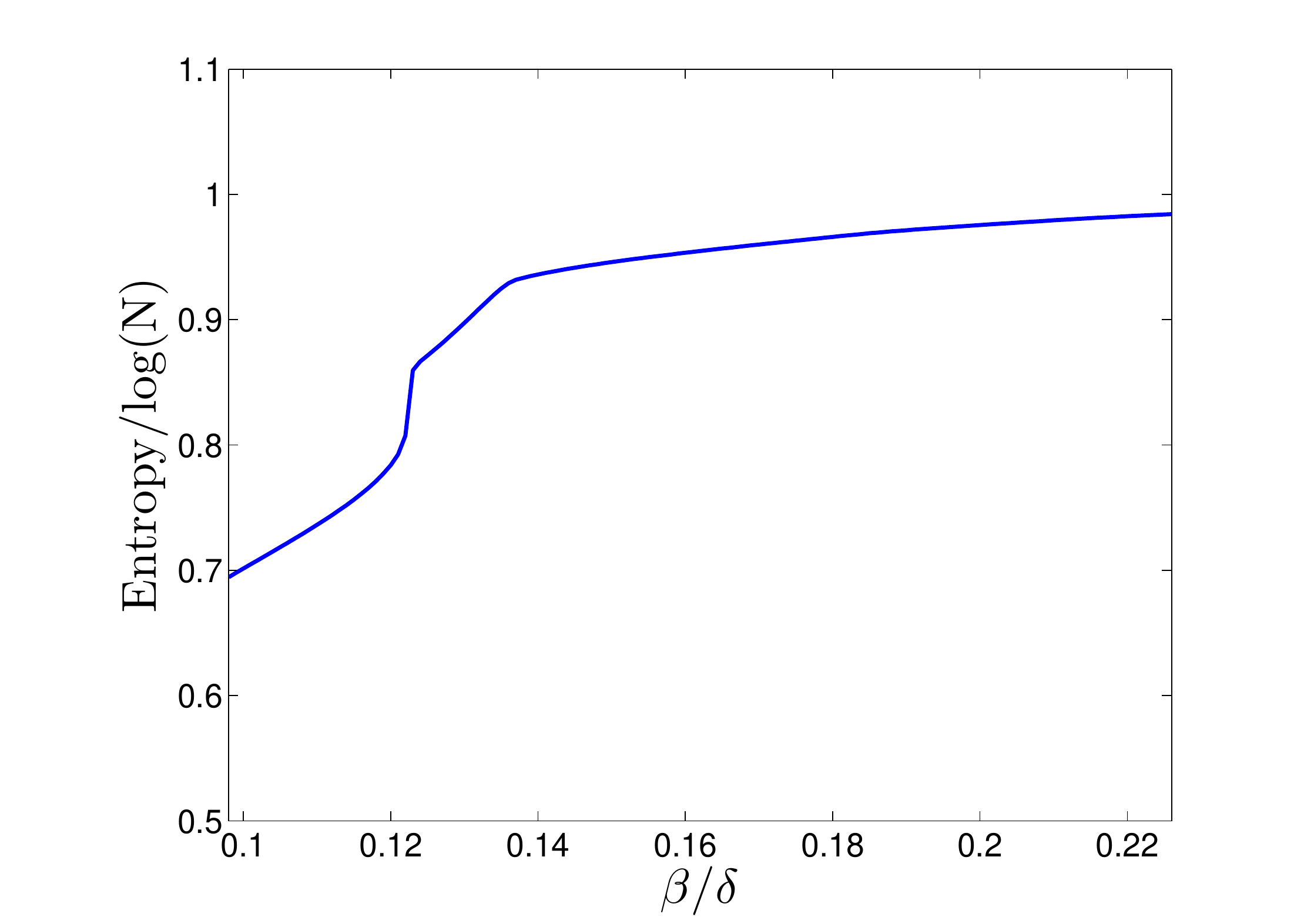}
\caption{}
\label{arbop}
\end{subfigure}
\begin{subfigure}{3.4in}
\includegraphics[width=3.8in]{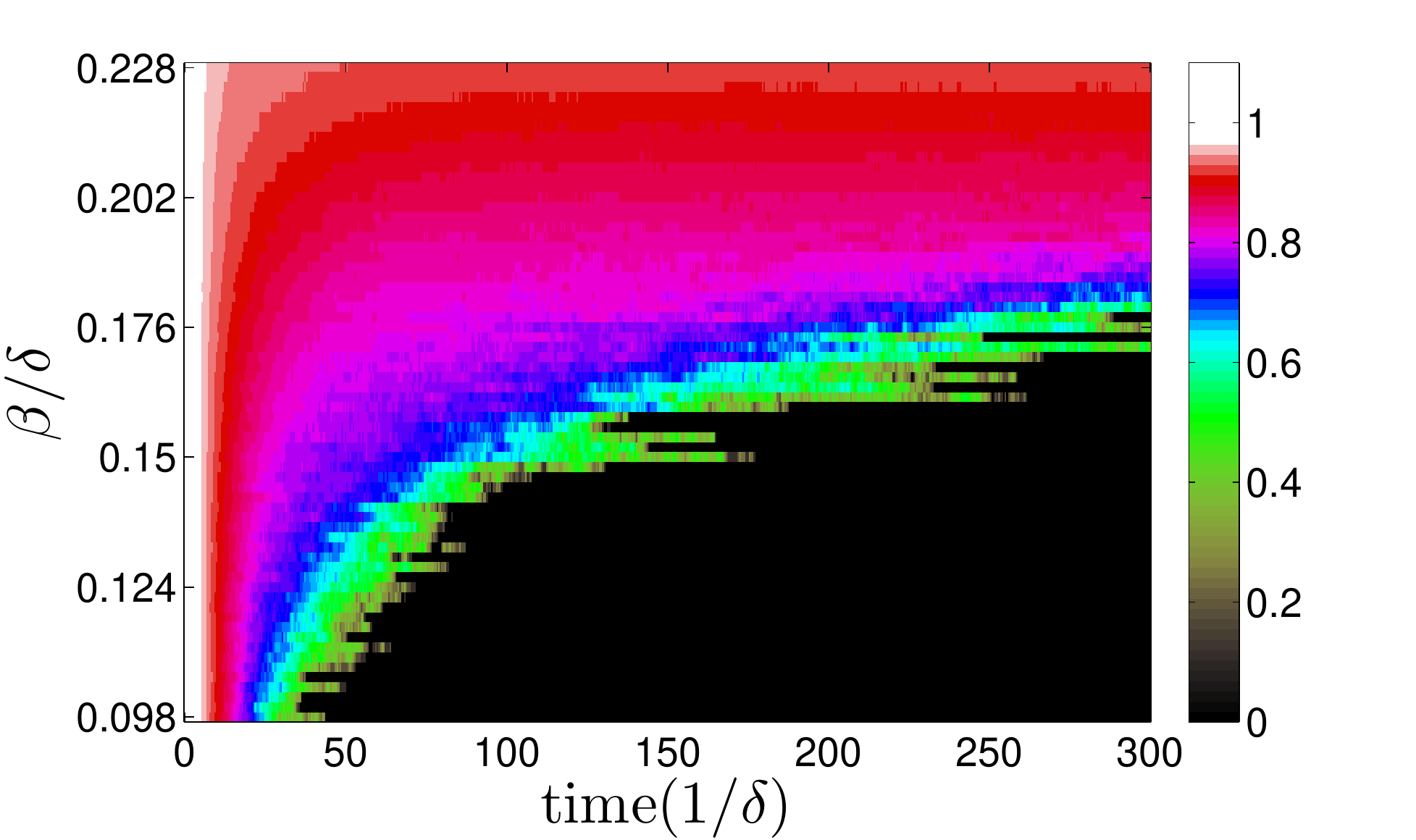}
\caption{}
\label{arbsim}
\end{subfigure}
\caption{(a) The entropy of the optimized distribution normalized by $\ln(N)$ for coauthorships network of Fig. \ref{arb1}. (b) Monte Carlo simulation for the SIS where all the nodes were initially infected. Color represents dispersion entropy of the infection probability distribution divided by $\ln(N)$.}
\label{arb2}%
\end{figure}

\section{conclusion}
In summary, we investigated the infection localization of SIS process. We used dispersion entropy defined in Eq. (\ref{a4}) as a measure of delocalization. We believe, in addition to infection size, measures such as dispersion entropy are relevant in epidemic spreading processes and should be included in numerical simulations. Moreover, we found an upper bound for the infection dispersion entropy when a metastable state exist. This upper bound, which depends on the infection strength, suggests the maximum number of nodes that can be active in a metastable state. A small upper bound for the dispersion entropy of a metastable state provides a sufficient condition for either complete extinction of infections or their localized persistence.

We would like to emphasize that studying the delocalization of infection using the endemic equilibrium point of the NIMFA model, i.e.,  the solution of Eq. (\ref{NIMFA_eq}), is not supported by rigorous arguments. This is mainly due to the fact that NIMFA equilibrium does not necessarily correspond to the actual metastable state for a generic graph. Our proposed upper bound obtained is larger than the entropy of NIMFA equilibrium point because the search space of the optimization includes the solution of Eq. \ref{NIMFA_eq}.

Finally, the presented optimization approach to the delocalization problem depends on a notion of the metastable state where the derivation of infection probability is close to zero. Since the disease-free state is an absorbing state for SIS process, this derivative is indeed very small yet negative. This necessitates further research investigating the sensitivity of the maximum enthropy problem to small perturbations of the search space boundary. Specifically, one should study the behavior of the optimal value $S^*$ subject to $(\tau A-I)\boldsymbol{p}+\epsilon \boldsymbol{r} \geq 0$ as $\epsilon \downarrow 0$ given a positive vector $\boldsymbol{r}$.
\bibliographystyle{IEEEtran}
\bibliography{Refrence}

\begin{thebibliography}{10}
\providecommand{\url}[1]{#1}
\csname url@samestyle\endcsname
\providecommand{\newblock}{\relax}
\providecommand{\bibinfo}[2]{#2}
\providecommand{\BIBentrySTDinterwordspacing}{\spaceskip=0pt\relax}
\providecommand{\BIBentryALTinterwordstretchfactor}{4}
\providecommand{\BIBentryALTinterwordspacing}{\spaceskip=\fontdimen2\font plus
\BIBentryALTinterwordstretchfactor\fontdimen3\font minus
  \fontdimen4\font\relax}
\providecommand{\BIBforeignlanguage}[2]{{%
\expandafter\ifx\csname l@#1\endcsname\relax
\typeout{** WARNING: IEEEtran.bst: No hyphenation pattern has been}%
\typeout{** loaded for the language `#1'. Using the pattern for}%
\typeout{** the default language instead.}%
\else
\language=\csname l@#1\endcsname
\fi
#2}}
\providecommand{\BIBdecl}{\relax}
\BIBdecl

\bibitem{wang2003epidemic}
Y.~Wang, D.~Chakrabarti, C.~Wang, and C.~Faloutsos, ``Epidemic spreading in
  real networks: An eigenvalue viewpoint,'' in \emph{Reliable Distributed
  Systems, 2003. Proceedings. 22nd International Symposium on}.\hskip 1em plus
  0.5em minus 0.4em\relax IEEE, 2003, pp. 25--34.

\bibitem{daley2001epidemic}
D.~J. Daley, J.~Gani, and J.~M. Gani, \emph{Epidemic modelling: an
  introduction}.\hskip 1em plus 0.5em minus 0.4em\relax Cambridge University
  Press, 2001.

\bibitem{kephart1991directed}
J.~O. Kephart and S.~R. White, ``Directed-graph epidemiological models of
  computer viruses,'' in \emph{Research in Security and Privacy, 1991.
  Proceedings., 1991 IEEE Computer Society Symposium on}.\hskip 1em plus 0.5em
  minus 0.4em\relax IEEE, 1991, pp. 343--359.

\bibitem{ganesh2005effect}
A.~Ganesh, L.~Massouli{\'e}, and D.~Towsley, ``The effect of network topology
  on the spread of epidemics,'' in \emph{INFOCOM 2005. 24th Annual Joint
  Conference of the IEEE Computer and Communications Societies. Proceedings
  IEEE}, vol.~2.\hskip 1em plus 0.5em minus 0.4em\relax IEEE, 2005, pp.
  1455--1466.

\bibitem{van2009virus}
P.~Van~Mieghem, J.~Omic, and R.~Kooij, ``Virus spread in networks,''
  \emph{Networking, IEEE/ACM Transactions on}, vol.~17, no.~1, pp. 1--14, 2009.

\bibitem{mountford2013metastable}
T.~Mountford, D.~Valesin, and Q.~Yao, ``Metastable densities for the contact
  process on power law random graphs,'' \emph{Electron. J. Probab}, vol.~18,
  no. 103, pp. 1--36, 2013.

\bibitem{goltsev2012localization}
A.~V. Goltsev, S.~N. Dorogovtsev, J.~Oliveira, and J.~F. Mendes, ``Localization
  and spreading of diseases in complex networks,'' \emph{Physical review
  letters}, vol. 109, no.~12, p. 128702, 2012.

\bibitem{boguna2013nature}
M.~Bogu{\~n}{\'a}, C.~Castellano, and R.~Pastor-Satorras, ``Nature of the
  epidemic threshold for the susceptible-infected-susceptible dynamics in
  networks,'' \emph{Physical review letters}, vol. 111, no.~6, p. 068701, 2013.

\bibitem{odor2013spectral}
G.~{\'O}dor, ``Spectral analysis and slow spreading dynamics on complex
  networks,'' \emph{Physical Review E}, vol.~88, no.~3, p. 032109, 2013.

\bibitem{kullback1951information}
S.~Kullback and R.~A. Leibler, ``On information and sufficiency,'' \emph{The
  annals of mathematical statistics}, pp. 79--86, 1951.

\bibitem{boyd2004convex}
S.~Boyd and L.~Vandenberghe, \emph{Convex optimization}.\hskip 1em plus 0.5em
  minus 0.4em\relax Cambridge university press, 2004.

\bibitem{newman2010networks}
M.~Newman, \emph{Networks: an introduction}.\hskip 1em plus 0.5em minus
  0.4em\relax Oxford University Press, 2010.

\bibitem{van2010graph}
P.~Van~Mieghem, \emph{Graph spectra for complex networks}.\hskip 1em plus 0.5em
  minus 0.4em\relax Cambridge University Press, 2010.

\bibitem{sahneh2013generalized}
F.~D. Sahneh, C.~Scoglio, and P.~Van~Mieghem, ``Generalized epidemic mean-field
  model for spreading processes over multilayer complex networks,''
  \emph{Networking, IEEE/ACM Transactions on}, vol.~21, no.~5, pp. 1609--1620,
  2013.

\bibitem{cator2012second}
E.~Cator and P.~Van~Mieghem, ``Second-order mean-field
  susceptible-infected-susceptible epidemic threshold,'' \emph{Physical review
  E}, vol.~85, no.~5, p. 056111, 2012.

\bibitem{van2011n}
P.~Van~Mieghem, ``The {N}-intertwined {SIS} epidemic network model,''
  \emph{Computing}, vol.~93, no. 2-4, pp. 147--169, 2011.

\bibitem{van2012n}
------, ``N-interwitned mean-field approximation (nimfa) versus exatc sis
  epidemics on networks,'' 2012.

\bibitem{van2012epidemic}
------, ``Epidemic phase transition of the sis type in networks,'' \emph{EPL
  (Europhysics Letters)}, vol.~97, no.~4, p. 48004, 2012.

\bibitem{li2012susceptible}
C.~Li, R.~van~de Bovenkamp, and P.~Van~Mieghem,
  ``Susceptible-infected-susceptible model: A comparison of {N}-intertwined and
  heterogeneous mean-field approximations,'' \emph{Physical Review E}, vol.~86,
  no.~2, p. 026116, 2012.

\bibitem{khanafer2014stability}
A.~Khanafer, T.~Basar, and B.~Gharesifard, ``Stability properties of infection
  diffusion dynamics over directed networks,'' in \emph{Decision and Control
  (CDC), 2014 IEEE 53rd Annual Conference on}.\hskip 1em plus 0.5em minus
  0.4em\relax IEEE, 2014, pp. 6215--6220.

\bibitem{van2013non}
P.~Van~Mieghem and R.~Van~de Bovenkamp, ``Non-markovian infection spread
  dramatically alters the susceptible-infected-susceptible epidemic threshold
  in networks,'' \emph{Physical review letters}, vol. 110, no.~10, p. 108701,
  2013.

\bibitem{cvx}
I.~CVX~Research, ``{CVX}: Matlab software for disciplined convex programming,
  version 2.0,'' \url{http://cvxr.com/cvx}, 2012.

\bibitem{GEMFsim}
F.~D. Sahneh, A.~Vajdi, H.~Shakeri, F.~Futing, and C.~Scoglio, ``{GEMFsim}: A
  stochastic simulator for the generalized epidemic modeling framework,''
  \emph{ar{X}iv}, 2016.

\bibitem{newman2006finding}
M.~E. Newman, ``Finding community structure in networks using the eigenvectors
  of matrices,'' \emph{Physical review E}, vol.~74, no.~3, p. 036104, 2006.

\end{thebibliography}

\end{document}